\documentclass[11pt,a4paper]{article}
\usepackage[margin=2.25cm,headheight=15pt]{geometry}
\usepackage[T1]{fontenc}
\usepackage{lmodern}
\usepackage{amsmath,amssymb,amsthm,bm,mathtools}
\usepackage{graphicx,booktabs,array,enumitem,fancyhdr,caption}
\usepackage[hidelinks]{hyperref}
\setlist{nosep,leftmargin=1.5em}
\newtheorem{theorem}{Theorem}[section]

\newtheorem{corollary}[theorem]{Corollary}
\theoremstyle{definition}

\newcommand{\E}{\mathbb E}

\newcommand{\R}{\mathbb R}
\newcommand{\tr}{\operatorname{tr}}

\newcommand{\ph}{\operatorname{ph}}
\newcommand{\ii}{\mathrm i}
\newcommand{\normunit}{\operatorname{n}}
\newcommand{\CN}{\mathcal{CN}}

\hypersetup{pdftitle={Optimal Receive Design on the Observation Axis for RSMA Satellite-UAV Downlinks},pdfauthor={Xianwei Meng}}
\begin{document}
\begin{center}
{\LARGE\bfseries Optimal Receive Design on the Observation Axis\\[.3em]for RSMA Satellite--UAV Downlinks}\par
\vspace{.65em}
{\large Xianwei Meng}\par
Hefei University of Technology, Hefei, China
\end{center}
\vspace{.4em}
\begin{abstract}
A receiver can change its observation state while the physical signal remains the same. In rate-splitting multiple access (RSMA) downlinks, this distinction separates receive-beam design and phase-reference transport from the motion of the satellite--UAV line of sight. We formulate receive optimization within the physical--observation dual-axis (PODA) structure and obtain three complementary results. For a known line-of-sight response and positive-definite receiver-noise covariance, a noise-weighted matched combiner simultaneously maximizes every decoding SINR and the minimum user rate. Under directional uncertainty, the principal eigenvector of the conditional response covariance globally maximizes mean array gain; for a Gaussian directional model, an equal-amplitude receiver is optimized by steering towards the conditional mean. Mean rate requires a different objective. We prove that the RSMA common-rate water level is jointly concave in the users' receive gains and derive a semidefinite relaxation with a computable global upper bound. A feasible rank-one beam can therefore be compared directly with the optimum of the finite-scenario problem. We also derive a precision-weighted linear predictor that suppresses uncertain physical-velocity estimates after observation-frame transport. Paired dynamic simulations reduce direction-cosine RMSE by 51.29\% relative to unrestricted nine-point extrapolation. In a separate uncertainty experiment, direct rate optimization increases mean minimum spectral efficiency from 0.32425 to 0.41712 bit/s/Hz, an improvement of 28.64\%; the corresponding design objective is within $2.28\times10^{-4}$ bit/s/Hz of a global upper bound. The phase-transport equation fixes a coherent representative of each optimal beam without changing its gain.
\end{abstract}

\section{Introduction}
Directional satellite downlinks depend on both the propagation geometry and the state of the receiving array. Satellite motion changes the physical line of sight. UAV rotation changes the coordinates in which that line of sight is observed. Receive weights and electronic references introduce a further dependence: two combinations of the same array snapshot may have different complex responses. These changes enter one recorded signal, although they have different causes and require different operations.

Rate-splitting multiple access separates a common stream from user-specific private streams and allows interference to be partly decoded and partly treated as noise \cite{mao}. The resulting waveform can also support sensing. Satellite RSMA studies have considered communication--sensing trade-offs and fair resource allocation \cite{yin,park}. The receiver-side formulation of Xu et al. \cite{source} uses decoded common and intended-private symbols as known components for angle-of-arrival estimation, while the spatial covariance of the remaining private streams supplies additional angular information. The same formulation relates the estimation bound to array gain and minimum effective user rate.

A question remains at the receiver: given its information and the available hardware, which observation state should it choose? A direction estimate alone does not answer this question. Steering towards an estimated angle may maximize gain for that estimate but need not maximize average rate under uncertainty. Likewise, extrapolating a fitted velocity may remove a deterministic delay while amplifying measurement noise. A coherent receiver must also preserve the phase convention when its weights change.

The physical--observation dual-axis (PODA) structure of Physobser Theory represents the physical state and the conditions that form an observation as separate arguments of the response law~\cite{meng_dual_axis}. Here the physical state contains the source, line of sight and transmitted symbols; the observation state contains the array frame, receive weights, calibration and phase reference. The associated phase-transport theory describes phase along joint-state paths; its fixed-physical-state restriction compares coherent responses to the same source~\cite{meng_phase_transport}. Receive optimization selects among these observation states using the information available to the receiver.

This paper develops that formulation for the single-satellite line-of-sight model. The common and private streams share one receive-array response at each UAV, so the known-response problem reduces to a generalized Rayleigh quotient. Directional uncertainty leads to a conditional covariance eigenproblem for mean gain. The corresponding mean-rate problem is more demanding because the logarithmic decoding rates and the common-stream bottleneck act on each realization before averaging. We establish the concavity of the fair-rate water level and use it to construct an upper bound over all rank-one digital receive beams.

Attitude-assisted angular prediction is established in UAV communications \cite{wang}. The present analysis uses this geometry while distinguishing observation-frame transport from physical prediction. Its contribution lies in their combination with conditional receive optimization, coherent phase transport, and explicit bounds on the achievable rate. The numerical study includes a paired comparison with the preceding fixed-window implementation and a separate experiment that optimizes the rate objective itself. All mathematical results are stated for their corresponding information set, loss and hardware constraint.

\section{Signal Model and Admissible Observation States}
\subsection{Physical and observation variables}
Following the PODA formulation in~\cite{meng_dual_axis}, let $X$ and $S$ be the physical and observation state spaces and $B\subseteq X\times S$ their realizable relation. The response law is
\begin{equation}
F:B\longrightarrow Y.
\end{equation}
The physical state $x\in X$ specifies the source and propagation variables. The observation state $s\in S$ specifies how they form a record. Fixing $s$ gives an ordinary physical signal model; fixing $x$ permits the comparison of observation states on the same source.

For UAV $k$, let $\ell_k\in S^2$ be the incoming unit direction in an inertial frame and $R_k\in\mathrm{SO}(3)$ the body-to-inertial rotation. The local direction cosines are
\begin{equation}
\eta_k=P R_k^T\ell_k,\qquad
P=\begin{bmatrix}1&0&0\\0&1&0\end{bmatrix}.
\label{eq:local}
\end{equation}
An $M_x\times M_y$ half-wavelength uniform planar array (UPA) has $N=M_xM_y$ elements. With centered element coordinates $r_m$, define
\begin{equation}
[a_r(\eta)]_m=e^{\ii\pi r_m^T\eta},\qquad
\bar a(\eta)=a_r(\eta)/\sqrt N.
\label{eq:steering}
\end{equation}
Thus $\|a_r\|^2=N$ and $\|\bar a\|=1$. The normalized vector will be used for receive gains.

The array signal is
\begin{equation}
y_k=a_k\left(b_{k,c}s_c+\sum_{j=1}^K b_{k,j}s_j\right)+n_k,
\qquad a_k=C_k a_r(\eta_k),
\label{eq:signal}
\end{equation}
where $C_k$ is the array calibration response, $b_{k,i}$ contains the fixed transmit beam and propagation coefficient, and $n_k\sim\CN(0,\mathcal N_k)$ with $\mathcal N_k\succ0$. The common stream $s_c$ and private streams $s_j$ have unit power and are mutually independent. A digital combiner produces $z_k=w_k^Hy_k$.

The receiver selects $w_k$ and its coherent reference. Actual UAV attitude is treated as a measured formation condition rather than a freely assignable control. All transmit coefficients and powers remain fixed when receive strategies are compared. Array-domain samples are available to the digital receiver. An equal-amplitude phase-shifter implementation is considered separately.

Let $\mathcal D$ denote the receiver's time-stamped history, nominal noise levels and available attitude information. When direction is uncertain, a chosen beam is a function of $\mathcal D$, not of the unobserved direction realization. A conditional distribution describes that uncertainty. Same-source phase transport compares two observation states at a common physical state~\cite{meng_phase_transport}. Averaging over possible physical states preserves this same-source requirement.

\subsection{Decoding rates and common-rate sharing}
For ideal calibration $C_k=I$ and white noise, set $\mathcal N_k=\sigma_k^2I$ and
\begin{equation}
g_k=|w_k^H\bar a(\eta_k)|^2,\quad \|w_k\|=1,
\qquad q_{k,i}=\frac{N|b_{k,i}|^2}{\sigma_k^2}.
\end{equation}
The normalized gain satisfies $0\le g_k\le1$. Define
\begin{align}
f_{c,k}(g)&=\log_2\left(1+\frac{q_{k,c}g}{1+g\sum_jq_{k,j}}\right),\\
f_{p,k}(g)&=\log_2\left(1+\frac{q_{k,k}g}{1+g\sum_{j\ne k}q_{k,j}}\right).
\label{eq:rates}
\end{align}
Perfect common-stream cancellation is assumed before private decoding. The total common rate is $r_c(g)=\min_k f_{c,k}(g_k)$ and the private rate is $r_k=f_{p,k}(g_k)$. Optimal common-rate sharing gives the water level $T(g)$ determined by
\begin{equation}
\sum_k(T-r_k)_+=r_c,\qquad C_k=(T-r_k)_+.
\label{eq:water}
\end{equation}
When $r_c=0$, take $T=\min_kr_k$. This is the minimum user spectral efficiency for the given gain realization. The stochastic objective in this paper is $\E[T(g)\mid\mathcal D]$, with common-rate allocation adapted to each realization. It is an ideal physical-layer rate measure; a fixed code rate and delayed feedback would define a different objective.

\section{Exact Optima at a Known Physical State}
\subsection{Noise-weighted receive matching}
All streams at UAV $k$ share $a_k$. Their decoding SINRs depend on $w_k$ through the single quantity
\begin{equation}
\kappa_k(w_k)=\frac{|w_k^Ha_k|^2}{w_k^H\mathcal N_kw_k}.
\end{equation}
For example,
\begin{equation}
\gamma_{c,k}=\frac{|b_{k,c}|^2\kappa_k}{1+\kappa_k\sum_j|b_{k,j}|^2}.
\end{equation}
The analogous private SINR is obtained by excluding the intended stream from the interference sum.

\begin{theorem}[Known-response observation optimum]
For a nonzero known response $a_k$ and $\mathcal N_k\succ0$, every nonzero combiner satisfies
\begin{equation}
\kappa_k(w_k)\le a_k^H\mathcal N_k^{-1}a_k.
\end{equation}
Equality holds exactly for
\begin{equation}
w_k^\star=c_k\mathcal N_k^{-1}a_k,\qquad c_k\ne0.
\label{eq:knownopt}
\end{equation}
The choices \eqref{eq:knownopt} simultaneously maximize all common and intended-private decoding SINRs, and hence maximize the fair rate for fixed transmit coefficients.
\end{theorem}
\begin{proof}
Cauchy--Schwarz gives
\[
|w^Ha|^2=|(\mathcal N^{1/2}w)^H\mathcal N^{-1/2}a|^2
\le(w^H\mathcal Nw)(a^H\mathcal N^{-1}a).
\]
Equality requires the two transformed vectors to be proportional. Each SINR is nondecreasing in $\kappa$, and every feasible common/private rate pair remains feasible when a user's gain increases. The optimum is therefore simultaneous across decoding stages and users.
\end{proof}
For white noise, ideal calibration $C_k=I$ and unit-norm weights, \eqref{eq:knownopt} becomes $w^\star=e^{\ii\phi}\bar a$. The ideal-alignment result is consequently a global upper limit for this physical model, rather than an additional receiver algorithm with available direction information.

\subsection{Phase transport selects the coherent representative}
The gain optimum leaves a scalar phase undetermined. A coherent choice of that phase follows from the response itself. Let $h(x,s)=w^Ha$ be smooth on the admissible joint-state domain and write $B^\times=\{(x,s)\in B:h(x,s)\ne0\}$. The normalized response and its phase one-form are
\begin{equation}
u=\frac{h}{|h|},\qquad
A_h=-\ii u^{-1}du
=\operatorname{Im}\frac{dh}{h}
=\frac{\operatorname{Im}(h^*dh)}{|h|^2}.
\label{eq:jointphaseform}
\end{equation}
Here $d$ differentiates on the joint domain. The denominator removes amplitude dependence. The form $A_h$ lives on $B^\times$; the corresponding principal connection lives on $P_{\rm ph}=B^\times\times U(1)$. With $\pi_{\rm ph}$ and $\mathrm{pr}_{U(1)}$ the two projections, and $\omega_{U(1)}=-\ii\zeta^{-1}d\zeta$, that connection is
\begin{equation}
\Omega_h=\mathrm{pr}_{U(1)}^*\omega_{U(1)}-\pi_{\rm ph}^*A_h.
\label{eq:principalconnection}
\end{equation}
For a piecewise-smooth joint-state path $c=(x,s)$ and its lift $\widetilde c=(c,\zeta)$, the intrinsic basic phase-transport equation (PTFE) is~\cite[Secs.~5.3--6.2]{meng_phase_transport}
\begin{equation}
\mathfrak E_h[\widetilde c]=\widetilde c^{\,*}\Omega_h=0.
\label{eq:intrinsicptfe}
\end{equation}
In a phase coordinate this is $\dot\zeta-\ii A_h(\dot c)\zeta=0$. The joint path is supplied by the physical motion and receiver operation; PTFE determines the phase carried along that path.

Now hold $x$ fixed and let $s(\lambda)$ remain in a nonzero response slice. The restriction of \eqref{eq:intrinsicptfe} to this path is
\begin{equation}
A_S=-\ii u^{-1}d_Su=\frac{\operatorname{Im}(h^*d_Sh)}{|h|^2},
\qquad \frac{d\tau}{d\lambda}=\ii A_S\!\left(\frac{ds}{d\lambda}\right)\tau,
\qquad \tau(0)=1.
\label{eq:ptfecont}
\end{equation}
The derivative $d_S$ varies only the observation state. For old and new responses $h_a=w_a^Ha$ and $h_b=w_b^Ha$, the finite solution is~\cite[Secs.~6.4--6.5]{meng_phase_transport}
\begin{equation}
\tau_{b\leftarrow a}(x)=\exp\!\left(\ii\int_{s_a}^{s_b}A_S\right)
=u_bu_a^{-1}=\frac{h_bh_a^*}{|h_b||h_a|}.
\label{eq:ptfe}
\end{equation}
This continuous transport requires a connecting path within the same nonzero slice. If the endpoints lie in different path components, the normalized ratio still converts their phases algebraically, but no such path realizes it. Since $u$ is globally defined on $B^\times$, $dA_h=0$: the induced connection is flat and circular transport around a closed path is unity, even when the unwrapped phase accumulates $2\pi n$. A response zero terminates this phase description and requires a new nonzero reference. These properties concern the scalar response-induced connection, rather than an arbitrary principal connection.

Writing $\ph(z)=z/|z|$ for a nonzero complex number, the finite solution satisfies $\ph(h_b)=\tau_{b\leftarrow a}\ph(h_a)$. Multiplying the new output by $\tau^{-1}$ gives
\begin{equation}
\tau^{-1}h_bm=|h_b|\ph(h_a)m.
\end{equation}
Equivalently, the coherent receive weight is
\begin{equation}
w_b^{\rm coh}=\tau_{b\leftarrow a}w_b.
\label{eq:weightphase}
\end{equation}
The weight multiplier and output multiplier have opposite signs because combining uses $w^H$. They are equivalent implementations of one reference change.

For a changing physical state, the comparison is made anew at each epoch. With a fixed reference state $s_r$, define~\cite[Sec.~6.6]{meng_phase_transport}
\begin{equation}
T_r(t)=u(x(t),s_r)u(x(t),s(t))^{-1}.
\end{equation}
Whenever both responses are nonzero, a noiseless received waveform $z(t)=h(x(t),s(t))m(t)$ obeys
\begin{equation}
T_r(t)z(t)=|h(x(t),s(t))|u(x(t),s_r)m(t).
\label{eq:dynamicreference}
\end{equation}
The source waveform and current amplitude are retained, while the array-response phase is expressed at the chosen reference. A continuous real phase $\phi_r$ of $h(x(t),s_r)$ on a smooth fixed-reference slice satisfies
\begin{equation}
\dot\phi_r(t)=\operatorname{Im}
\frac{D_xh(x(t),s_r)[\dot x(t)]}{h(x(t),s_r)}.
\label{eq:referencephase}
\end{equation}
On a smooth interval where $m(t)\ne0$, choose a continuous phase $\phi_m$. The corrected waveform then has phase rate $\dot\phi_r+\dot\phi_m$. The carrier phase of $m(t)$ is therefore recovered up to a constant when $u(x(t),s_r)$ is constant along the physical trajectory; a known varying reference phase can instead be calibrated out. More generally, identification of a transported observation with an intrinsic physical phase requires phase faithfulness or calibration of the reference observation. Equation~\eqref{eq:dynamicreference} itself establishes conversion to the reference state.

If $\tau$ is uncertain, its conditional circular mean supplies an exact estimation result. For a unit-modulus decision $v$, the conditional phase loss is
\begin{equation}
\E[|v^*\tau-1|^2\mid\mathcal D]
=2-2\operatorname{Re}\{v^*m_\tau\},\qquad
m_\tau=\E[\tau\mid\mathcal D].
\end{equation}
It is minimized by $v^\star=m_\tau/|m_\tau|$ when $m_\tau\ne0$; every phase ties when $m_\tau=0$. The inverse correction is $(v^\star)^*$. This result concerns circular chordal loss. 

For two combinations of the same raw snapshots, the correlation estimate can be made consistent by subtracting their common noise contribution:
\begin{equation}
\widehat c_{ba}=\frac1L\sum_{\ell=1}^Lz_b[\ell]z_a[\ell]^*
-w_b^H\mathcal Nw_a,\qquad \widehat\tau=\ph(\widehat c_{ba}).
\end{equation}
Here the two weights act at the same physical epoch. A ratio across different physical epochs also contains physical motion and symbol evolution. Phase transport fixes a coherent representative of a gain-optimal beam; it does not increase the power SINR by itself.

\section{Conditional Mean-Gain Optimization}
\subsection{The posterior response covariance}
For white noise, define the normalized conditional response covariance
\begin{equation}
Q_k=\E[\bar a(\eta_k)\bar a(\eta_k)^H\mid\mathcal D],
\qquad Q_k\succeq0,\quad\tr Q_k=1.
\end{equation}
The mean receive gain is $w_k^HQ_kw_k$.

\begin{theorem}[Conditional mean-gain optimum]
Among all unit-norm digital receive weights,
\begin{equation}
w_{G,k}^\star=v_{\max}(Q_k),\qquad
\overline g_k^\star=\lambda_{\max}(Q_k).
\label{eq:gainopt}
\end{equation}
If the largest eigenvalue is simple, the solution is unique up to its coherent phase.
\end{theorem}
\begin{proof}
Writing $Q=V\Lambda V^H$, any unit vector $w$ satisfies
$w^HQw=\sum_i\lambda_i|(V^Hw)_i|^2\le\lambda_{\max}$.
Equality holds precisely in the leading eigenspace.
\end{proof}
For colored noise, apply the same argument to
$\mathcal N_k^{-1/2}\E[a_ka_k^H\mid\mathcal D]\mathcal N_k^{-1/2}$ and transform its leading eigenvector back by $\mathcal N_k^{-1/2}$.

If $\eta\mid\mathcal D\sim\mathcal N(\mu,\Sigma)$, the characteristic function gives
\begin{equation}
Q_{mn}=\frac1N\exp\left\{\ii\pi(r_m-r_n)^T\mu
-\frac{\pi^2}{2}(r_m-r_n)^T\Sigma(r_m-r_n)\right\}.
\label{eq:gaussQ}
\end{equation}
Write $Q=N^{-1}D_\mu B_\Sigma D_\mu^H$, where $[D_\mu]_{mm}=e^{\ii\pi r_m^T\mu}$. Every entry of $B_\Sigma$ is strictly positive. Its leading eigenvector can therefore be chosen strictly positive and is unique after normalization. The optimal beam has the form
\begin{equation}
w_G^\star=D_\mu v_\Sigma.
\end{equation}
Its phases point towards the conditional mean, while its amplitudes depend on directional uncertainty. At $\Sigma=0$, $v_\Sigma$ is uniform and the solution reduces to matched steering. For independent direction-cosine errors on a rectangular UPA, $B_\Sigma=B_x\otimes B_y$, so two small eigenproblems replace one full-array decomposition.

Equation \eqref{eq:gaussQ} is exact for the stated Gaussian model. A posterior restricted to the visible direction-cosine disk, or one with several modes, is handled by the definition of $Q$ rather than by changing the eigenvalue argument.

\subsection{Equal-amplitude hardware}
\begin{theorem}[Phase-only mean-gain optimum]
Under the Gaussian model \eqref{eq:gaussQ} and the constraint $|w_m|=1/\sqrt N$, the global mean-gain optimum is
\begin{equation}
w_{\rm ph}^\star=e^{\ii\phi}\bar a(\mu).
\end{equation}
\end{theorem}
\begin{proof}
Write $w=N^{-1/2}D_\mu(e^{\ii\theta_1},\ldots,e^{\ii\theta_N})^T$. Then
\[
w^HQw=\frac1{N^2}\sum_{m,n}[B_\Sigma]_{mn}\cos(\theta_n-\theta_m)
\le\frac1{N^2}\sum_{m,n}[B_\Sigma]_{mn}.
\]
All coefficients are positive, and equality holds when every residual phase is the same.
\end{proof}
Thus additional amplitude freedom, not a different phase-only steering rule, creates the mean-gain improvement in this posterior model. The rate-optimal design below uses full digital amplitude and phase weights.

\section{Direct Optimization of Mean Fair Rate}
\subsection{Gain optimality and rate optimality}
Mean gain and mean rate are different objectives. For a single user with two orthogonal unit responses occurring with probabilities $0.6$ and $0.4$, the mean-gain beam assigns all energy to the first response. At total SNR $10$, its mean rate is $0.6\log_2(11)$. A beam allocating energy fractions $x$ and $1-x$ obtains
\begin{equation}
0.6\log_2(1+10x)+0.4\log_2(1+10(1-x)),
\end{equation}
whose maximum occurs at $x=0.62$. Distributing gain across uncertain responses can therefore improve the logarithmic rate even when it reduces mean gain.

For the multiuser problem, the following property makes the distinction tractable.
\begin{theorem}[Concavity of the RSMA fair-rate function]
For fixed nonnegative stream powers, $T(g)$ in \eqref{eq:water} is jointly concave and nondecreasing in each component of $g\in\R_+^K$. It admits the representation
\begin{equation}
T(g)=\min_{\varnothing\ne A\subseteq\{1,\ldots,K\}}
\frac{\min_j f_{c,j}(g_j)+\sum_{k\in A}f_{p,k}(g_k)}{|A|}.
\label{eq:subset}
\end{equation}
\end{theorem}
\begin{proof}
For a desired power $S$ and interference power $I$, each rate is
$f(g)=\log_2[(1+(I+S)g)/(1+Ig)]$ and satisfies
\begin{align}
f'(g)&=\frac{S}{\ln2(1+Ig)[1+(I+S)g]}\ge0,\\
f''(g)&=-\frac{S[(2I+S)+2I(I+S)g]}{\ln2(1+Ig)^2[1+(I+S)g]^2}\le0.
\end{align}
Sort the private rates as $r_{(1)}\le\cdots\le r_{(K)}$. Solving the water-level equation yields
$T=\min_{1\le m\le K}(r_c+\sum_{i=1}^m r_{(i)})/m$.
For any fixed cardinality, the smallest subset sum selects those lowest rates, proving \eqref{eq:subset}. Each branch of \eqref{eq:subset} is concave and coordinatewise nondecreasing. Their pointwise minimum retains both properties.
\end{proof}

\begin{corollary}[A global conditional rate bound]
For any receive weights chosen from $\mathcal D$,
\begin{equation}
\E[T(g)\mid\mathcal D]
\le T(\E[g\mid\mathcal D])
\le T(\lambda_{\max}(Q_1),\ldots,\lambda_{\max}(Q_K)).
\label{eq:jensen}
\end{equation}
\end{corollary}
The first inequality follows from concavity; the second follows from monotonicity and \eqref{eq:gainopt}. The bound is often tight near perfect alignment. Under broad uncertainty it can be loose because averaging conceals simultaneous weak-user events.

\subsection{Finite-scenario optimization}
Let $\eta_k^{(j)}$, $j=1,\ldots,J$, be a fixed set of joint directional samples from the conditional model. Define
\begin{equation}
A_k^{(j)}=\bar a(\eta_k^{(j)})\bar a(\eta_k^{(j)})^H,\qquad
W_k=w_kw_k^H.
\end{equation}
The sample-average objective is
\begin{equation}
F_J(W)=\frac1J\sum_{j=1}^JT(g^{(j)}),\qquad
 g_k^{(j)}=\tr(A_k^{(j)}W_k).
\label{eq:sampleF}
\end{equation}
A single digital receive beam requires $W_k\succeq0$, $\tr W_k=1$ and $\operatorname{rank}W_k=1$. Removing the rank constraint gives
\begin{equation}
\max_{W_k\succeq0,\,\tr W_k=1}F_J(W).
\label{eq:relax}
\end{equation}
By the preceding theorem, $F_J$ is concave in all $W_k$. Problem \eqref{eq:relax} is consequently a convex optimization problem in the usual maximization form \cite{boyd}. Its optimum bounds the rate of every admissible rank-one beam for the same samples. The relaxed matrix is used as an optimization variable, not as the response of a single combiner.

\subsection{A computable bound and a feasible beam}
At a feasible $W$, let $G_k$ be a block supergradient of $F_J$. Such a supergradient is obtained by selecting an active common-rate bottleneck and active private-rate subset in \eqref{eq:subset}. Set
\begin{equation}
\Delta(W)=\sum_k\left[\lambda_{\max}(G_k)-\tr(G_kW_k)\right].
\label{eq:gap}
\end{equation}
\begin{theorem}[Global upper bound from a supergradient]
Every feasible $W'$ in \eqref{eq:relax} satisfies
\begin{equation}
F_J(W')\le U(W):=F_J(W)+\Delta(W).
\label{eq:upper}
\end{equation}
For any feasible rank-one receive beams $\widehat w_k$ with achieved objective $L_J$, the unknown rank-one optimum lies in $[L_J,U(W)]$.
\end{theorem}
\begin{proof}
Concavity gives
$F_J(W')\le F_J(W)+\sum_k\tr[G_k(W_k'-W_k)]$.
A positive-semidefinite unit-trace matrix obeys $\tr(G_kW_k')\le\lambda_{\max}(G_k)$. Substitution proves the bound. Every rank-one candidate belongs to the relaxed feasible set, so its objective is a lower bound on the original rank-one optimum.
\end{proof}
This bound is available at any iterate. It does not depend on a stopping rule being interpreted as convergence. If $U-L_J\le\varepsilon$, the recovered beam is within $\varepsilon$ of the global optimum of the fixed finite-scenario problem.

The implementation uses a conditional-gradient iteration. For each user it finds a leading unit eigenvector $v_k$ of $G_k$ and takes the direction $v_kv_k^H-W_k$. A scalar line search maximizes $F_J$ along the feasible segment. The smallest upper bound encountered is retained. A rank-one beam is then extracted from the leading eigenvector of each $W_k$ and refined on the product of complex unit spheres. All reported achievable rates are evaluated using those vectors. The upper bound remains the one obtained from \eqref{eq:upper}, so the rate difference includes any loss caused by extracting a rank-one solution.

Independent samples evaluate performance beyond the design set. The bound \eqref{eq:upper} pertains to the fixed design objective; the separate evaluation estimates the expected performance of the resulting receiver under the stated distribution.

\section{Optimal Linear Prediction after Frame Transport}
\subsection{Transport and physical evolution}
For array directions, the same-source comparison of~\cite{meng_dual_axis} holds the physical line of sight fixed. Substituting $q_a=R_a^T\ell$ and $q_b=R_b^T\ell$ gives
\begin{equation}
q_b=R_b^TR_aq_a.
\label{eq:frametrans}
\end{equation}
Writing $R^T\dot R=[\omega]_\times$, the full direction evolution is
\begin{equation}
\dot q=R^T\dot\ell-[\omega]_\times q.
\end{equation}
The second term is observation-frame transport. The first term is physical motion. With a known front hemisphere,
$\iota(u,v)=(u,v,\sqrt{1-u^2-v^2})^T$ maps the measured direction cosines to a unit vector. Historical measurements are first expressed in the inertial frame:
\begin{equation}
y_i^{\ell}=\widehat R_i\iota(\widehat\eta_i).
\end{equation}
Prediction is then performed on their local Cartesian coordinates. The result is normalized and mapped into the array frame at the application epoch.

\subsection{Precision-weighted shrinkage}
Consider one local coordinate over a short history window:
\begin{equation}
y_i=\alpha+v t_i+\epsilon_i,\qquad
\E[v]=0,\quad\operatorname{Var}(v)=q,\quad
\E[\epsilon_i]=0,\quad\operatorname{Var}(\epsilon_i)=s_i^2.
\label{eq:predmodel}
\end{equation}
Assume $s_i^2>0$ and at least two distinct observation times, so $S_{tt}>0$. The level $\alpha$ is an unknown deterministic constant. The velocity and measurement errors are mutually uncorrelated; different measurement errors are uncorrelated. Define
\begin{align}
p_i&=s_i^{-2},&S_0&=\sum_i p_i,&\bar t&=\frac{\sum_i p_it_i}{S_0},\\
\bar y&=\frac{\sum_i p_iy_i}{S_0},&
S_{tt}&=\sum_i p_i(t_i-\bar t)^2,&
\widehat v&=\frac{\sum_i p_i(t_i-\bar t)y_i}{S_{tt}}.
\end{align}
Time is measured relative to the latest observation, so the target epoch has coordinate $\tau\ge0$.

\begin{theorem}[Best linear unbiased prediction]
Under \eqref{eq:predmodel}, the best linear unbiased predictor of $\alpha+v\tau$ is
\begin{equation}
\widehat\xi_\tau=\bar y+(\tau-\bar t)\kappa\widehat v,\qquad
\kappa=\frac{qS_{tt}}{1+qS_{tt}}.
\label{eq:blup}
\end{equation}
Its mean-square prediction error is
\begin{equation}
V_\tau=\frac1{S_0}+(\tau-\bar t)^2\frac{q}{1+qS_{tt}}.
\label{eq:blupvar}
\end{equation}
\end{theorem}
\begin{proof}
Weighted centering gives $\bar y=\alpha+v\bar t+\bar\epsilon$ and $\widehat v=v+e_v$, where
$\operatorname{Var}(\bar\epsilon)=1/S_0$, $\operatorname{Var}(e_v)=1/S_{tt}$ and $\operatorname{Cov}(\bar\epsilon,e_v)=0$. Among centered predictors $\bar y+(\tau-\bar t)c\widehat v$, the risk is
\[
\frac1{S_0}+(\tau-\bar t)^2\left[(1-c)^2q+\frac{c^2}{S_{tt}}\right].
\]
Its unique minimum is attained at $c=qS_{tt}/(1+qS_{tt})$. Any additional linear combination orthogonal, in the inverse-variance metric, to the constant and linear time columns contains no level or velocity component and adds nonnegative error variance. Thus the minimum is over all linear predictors unbiased for every $\alpha$. Substitution gives \eqref{eq:blupvar}.
\end{proof}
The theorem uses second moments only. Under joint Gaussian assumptions and a diffuse prior for the level, the same expression is the conditional mean. When velocity information is weak, $\kappa$ tends to zero and the predictor approaches the weighted history mean. When velocity is well resolved, $\kappa$ approaches one and ordinary linear extrapolation is recovered.

Shrinking only an additional $\tau\widehat v$ term while leaving the fitted endpoint unchanged does not produce \eqref{eq:blup}. The endpoint already contains velocity-estimation noise. The entire horizon $\tau-\bar t$ must use the same shrinkage.

\subsection{Application to the array}
For the numerical receiver, each coordinate uses the nominal local variance
\begin{equation}
s_i^2=v_{\eta,i}+\sigma_R^2,\qquad
q=\frac12\left(\frac{7500}{599000}\right)^2,
\label{eq:nominal}
\end{equation}
where $v_{\eta,i}$ is the supplied angular-estimation variance and $\sigma_R$ is the attitude measurement standard deviation in radians. The velocity moment is set from the declared satellite-speed scale and vertical separation, with an isotropic horizontal heading. It is fixed before evaluating the trajectories.

After applying \eqref{eq:blup} to the three local coordinates, the receiver forms
\begin{equation}
\widehat\eta_\tau=P\widehat R_\tau^T\normunit(\widehat\xi_\tau),\qquad
\normunit(v)=v/\|v\|.
\end{equation}
The nominal directional covariance used for beam selection is $(V_\tau+\sigma_R^2)I_2$. The leading eigenvector of its conditional response covariance supplies the mean-gain beam. The current attitude uncertainty remains in this expression; historical averaging cannot remove uncertainty in the current frame.

The linear prediction theorem applies to the local moment model. Unit-vector normalization and exact rotations subsequently implement the geometry. Their performance is measured against the nonlinear trajectories in the next section. The receiver uses historical samples and application-time attitude measurements; future true directions are used only to evaluate its error.

\section{Numerical Results}
\subsection{Dynamic comparison with identical observations}
The dynamic experiment uses four UAVs and half-wavelength $8\times8$ transmit and receive UPAs. Satellite and UAV altitudes are 600 km and 1 km, respectively. UAV horizontal positions are uniform in area within a 25 km disk. The initial satellite position is uniform over $[-20,20]\times[-10,10]$ km in the horizontal plane; its speed is 7.5 km/s with a uniformly distributed horizontal heading. This local straight-line trajectory describes the short simulated interval. UAV positions are fixed, while roll and pitch are sinusoidal at 1.7 and 2.3 Hz, with independently randomized initial phases.

There are 128 independent trajectories, each with 300 evaluated epochs separated by 10 ms. Eight earlier epochs initialize the nine-point history window. The application age is 80 ms. In the main experiment, each attitude sinusoid has a peak amplitude of $3^\circ$, and the independent components of attitude measurement error have standard deviation $0.1^\circ$. Historical and application-time attitude measurements are both available; their errors are independent. The front hemisphere is known.

Angular observations use independent Gaussian errors with the direction-cosine CRLB covariance of Appendix B and $L=256$ snapshots. This is an efficient-estimation approximation with successfully reconstructed reference streams and perfect common-stream cancellation. The experiment examines prediction and beam selection at the specified measurement accuracy. Nominal angular error variances are supplied to the predictor as calibrated receiver statistics.

All receivers share the same propagation realizations, angular observations, attitude measurements and transmit resources. The common transmit beam is the normalized sum of the user steering vectors. Private beams are the normalized columns of
\begin{equation}
A_t(A_t^HA_t+\lambda I)^{-1},\qquad \lambda=K/10^{\mathrm{SNR}_{\rm dB}/10},
\end{equation}
where $A_t$ contains the normalized departure steering vectors. Common-stream power is 75\% of total power; each private stream receives 6.25\%. At a 20 GHz carrier, 15 MHz bandwidth, 300 K receiver-noise temperature and 3 dBi element gains, the total transmit power is 11.73559 W. It sets the ensemble mean total aligned receive SNR to 10 dB in the free-space model. These transmit beams and powers define a common reference for the receive comparison; the transmit optimization of \cite{source} is a separate problem.

Six receive strategies are compared. \emph{Stale AoA} holds the latest local angular observation. \emph{Unrestricted extrapolation} is the preceding dual-axis receiver: it transforms the nine observations into the inertial frame, fits an unweighted line, and extrapolates the full fitted velocity. \emph{Frame transport only} uses the fitted latest-epoch intercept without additional physical displacement. \emph{BLUP steering} applies \eqref{eq:blup} and steers towards the resulting direction. \emph{BLUP with eigenbeam} uses the same predicted mean with the Gaussian mean-gain optimum \eqref{eq:gainopt}. \emph{Ideal alignment} uses the true application-time direction and gives the known-state upper limit.

Table~\ref{tab:dynamic} reports mean minimum spectral efficiency, mean normalized receive gain, and two-coordinate direction-cosine RMSE,
\begin{equation}
\operatorname{RMSE}_{\eta}=\sqrt{\E\|\widehat\eta-\eta\|_2^2}.
\end{equation}
The new predictor reduces RMSE from 0.00641638 to 0.00312557, a 51.29\% reduction. The mean shrinkage factor is 0.06311: the physical slope is poorly resolved over the short window, so unrestricted extrapolation introduces appreciable noise. The mean minimum rate rises from 0.88313736 to 0.88361873 bit/s/Hz, or 0.05451\%. The modest rate increment is consistent with the preceding receiver already operating close to ideal alignment.

\begin{table}[tb]
\centering\small
\caption{Dynamic receive comparison: 10 dB aligned SNR, $3^\circ$ attitude peak, 80 ms age and $0.1^\circ$ attitude measurement noise. Rates are in bit/s/Hz.}
\label{tab:dynamic}
\begin{tabular}{lrrr}\toprule
Receiver & Mean minimum rate & Mean gain & Direction RMSE\\\midrule
Stale AoA & 0.83922443 & 0.87572648 & 0.05079493\\
Unrestricted extrapolation & 0.88313736 & 0.99787060 & 0.00641638\\
Frame transport only & 0.88357875 & 0.99935847 & 0.00351963\\
BLUP steering & 0.88361872 & 0.99949402 & 0.00312557\\
BLUP with eigenbeam & 0.88361873 & 0.99949404 & 0.00312557\\
Ideal alignment & 0.88376685 & 1.00000000 & 0\\\bottomrule
\end{tabular}
\end{table}

The paired rate difference between BLUP with eigenbeam and unrestricted extrapolation is $4.81372\times10^{-4}$ bit/s/Hz. Its approximate 95\% confidence interval is $[4.66147,4.96597]\times10^{-4}$, calculated from the 128 independent trajectory means. Individual time slots within a trajectory are not treated as independent replicates. Relative to stale AoA, the mean rate improves by 5.28992\%; the paired difference is 0.04439430 bit/s/Hz with interval $[0.04333755,0.04545106]$. The eigenbeam contributes only $7.17\times10^{-9}$ bit/s/Hz beyond BLUP steering in this narrow-posterior case. Almost all of the additional benefit therefore comes from the predictor.

At a $5^\circ$ attitude peak, stale AoA achieves 0.74384739 bit/s/Hz, unrestricted extrapolation achieves 0.88313823, and BLUP with eigenbeam achieves 0.88361907. The new receiver improves on stale AoA by 18.79037\%, while its increment over unrestricted extrapolation remains 0.05445\%. With a stationary satellite and zero attitude oscillation, the corresponding stale, unrestricted and BLUP-with-eigenbeam rates are 0.88372605, 0.88321511 and 0.88372747. The same nonzero velocity moment is retained in this stationary comparison. Shrinkage suppresses the noise of an unnecessary slope estimate.

\subsection{Sensitivity to attitude measurement accuracy}
Figure~\ref{fig:dynamic} and Table~\ref{tab:attnoise} increase attitude measurement noise while keeping the $3^\circ$ physical attitude peak, 80 ms age and all transmit resources fixed. At $2^\circ$ measurement noise, unrestricted extrapolation yields 0.66143711 bit/s/Hz. BLUP steering yields 0.82484480, and BLUP with eigenbeam yields 0.82611396. The improvement over unrestricted extrapolation is 24.89682\%. The remaining current-frame uncertainty limits performance even after historical prediction is improved. In this case, stale AoA still achieves 0.83922443 bit/s/Hz because it does not use the inaccurate attitude measurements; neither attitude-assisted method is uniformly preferable across all telemetry accuracies.

\begin{table}[tb]
\centering\small
\caption{Mean minimum rate under increasing attitude measurement noise; all other settings follow Table~\ref{tab:dynamic}.}
\label{tab:attnoise}
\begin{tabular}{rrrrr}\toprule
Noise std. ($^\circ$) & Unrestricted & BLUP steering & BLUP eigenbeam & Increase (\%)\\\midrule
0.1 & 0.883137 & 0.883619 & 0.883619 & 0.0545\\
0.5 & 0.874402 & 0.881072 & 0.881075 & 0.7632\\
1.0 & 0.838925 & 0.871982 & 0.872036 & 3.9468\\
2.0 & 0.661437 & 0.824845 & 0.826114 & 24.8968\\\bottomrule
\end{tabular}
\end{table}

\begin{figure}[tb]
\centering\includegraphics[width=\linewidth]{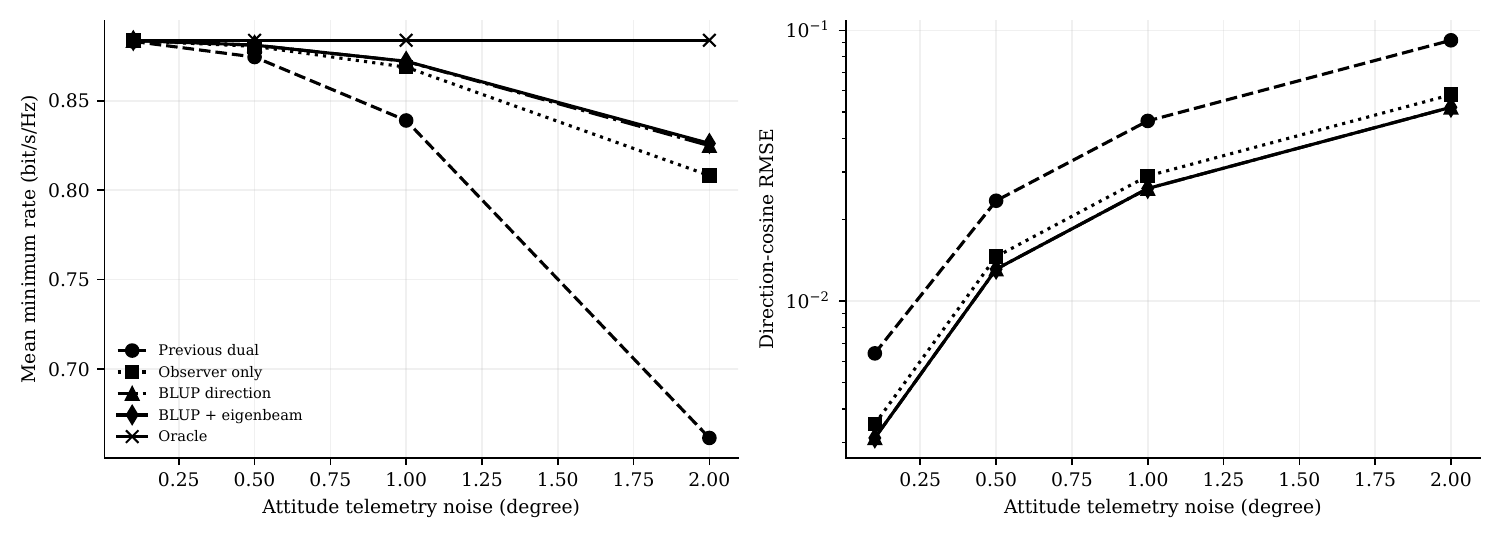}
\caption{Effect of attitude measurement noise on the dynamic receive comparison. ``Previous dual'' denotes unrestricted inertial extrapolation. The BLUP steering and eigenbeam curves nearly coincide at low noise. Ideal alignment has zero direction error and is omitted from the logarithmic error panel.}
\label{fig:dynamic}
\end{figure}

Performance is insensitive to the velocity-moment scale over the tested range. Multiplying the predeclared $q$ by 0.25 and 4 gives mean rates of 0.88361758 and 0.88360653 bit/s/Hz, respectively, compared with 0.88361873 at the nominal value. These alternatives retain the improvement over unrestricted extrapolation. The reported default remains the moment specified in \eqref{eq:nominal}.

\subsection{Direct rate optimization under directional uncertainty}
A separate experiment isolates the choice of digital receive beam under a supplied directional posterior. It uses four independent user directions around array boresight and the same $8\times8$ receive array. For each user, the direction-cosine error is an isotropic zero-mean Gaussian with scale $\sigma_\eta$, conditioned on the physical disk $u^2+v^2<1$. Thus $\sigma_\eta$ is a dimensionless Gaussian scale, rather than an attitude angle. Conditioning has negligible effect at $\sigma_\eta\le0.15$, but ensures that every sampled direction is realizable.

The common and private aligned SNR coefficients at every user are
\begin{equation}
(q_{k,c},q_{k,1},q_{k,2},q_{k,3},q_{k,4})
=(7.5,0.625,0.625,0.625,0.625).
\end{equation}
They represent an overlapping line-of-sight downlink with total aligned SNR 10 dB and the same 75\% common-stream power fraction. These coefficients are held fixed throughout the comparison. This posterior experiment does not reuse the distinct departure geometries and regularized transmit beams of the trajectory experiment.

A scrambled Sobol set of $J=512$ joint scenarios defines the design objective. A separately scrambled set of 8192 scenarios evaluates each final receive beam. The random seeds are 20260916 and 20261916, respectively. The three designs are center-matched steering, the principal eigenvector of the design-sample response covariance, and direct optimization of mean minimum rate. The eigenbeam exactly maximizes the mean gain on the design scenarios. Both optimized designs use the same 512 scenarios, and all three use unit-norm complex digital weights.

The semidefinite iteration uses at most 300 updates and a nominal supergradient-gap tolerance of $2\times10^{-5}$ bit/s/Hz. After leading-eigenvector extraction, at most 120 projected updates refine the rank-one beams. Table~\ref{tab:rates} reports their rates on the separate evaluation set. Figure~\ref{fig:rate} compares the achieved rate and mean gain.

\begin{table}[tb]
\centering\small
\caption{Mean minimum rate on 8192 evaluation scenarios. The increase is relative to center-matched steering; rates are in bit/s/Hz.}
\label{tab:rates}
\begin{tabular}{rrrrr}\toprule
$\sigma_\eta$ & Center matched & Mean-gain eigenbeam & Rate optimized & Increase (\%)\\\midrule
0.00 & 0.696812 & 0.696812 & 0.696812 & 0.000\\
0.03 & 0.672958 & 0.673272 & 0.673334 & 0.056\\
0.06 & 0.578822 & 0.586115 & 0.590006 & 1.932\\
0.10 & 0.324249 & 0.365670 & 0.417121 & 28.642\\
0.15 & 0.111900 & 0.142828 & 0.266322 & 138.001\\
0.20 & 0.044104 & 0.051012 & 0.180507 & 309.280\\\bottomrule
\end{tabular}
\end{table}

At $\sigma_\eta=0.10$, direct rate optimization increases mean minimum rate from 0.32424863 to 0.41712132 bit/s/Hz, an improvement of 28.64243\%. It also improves on the mean-gain eigenbeam by 14.07031\%. The rate-optimized design accepts a lower mean array gain in exchange for better gain distribution across the uncertain directions and weakest-user events. This is the multiuser counterpart of the distinction established in Section 5.1.

\begin{figure}[tb]
\centering\includegraphics[width=\linewidth]{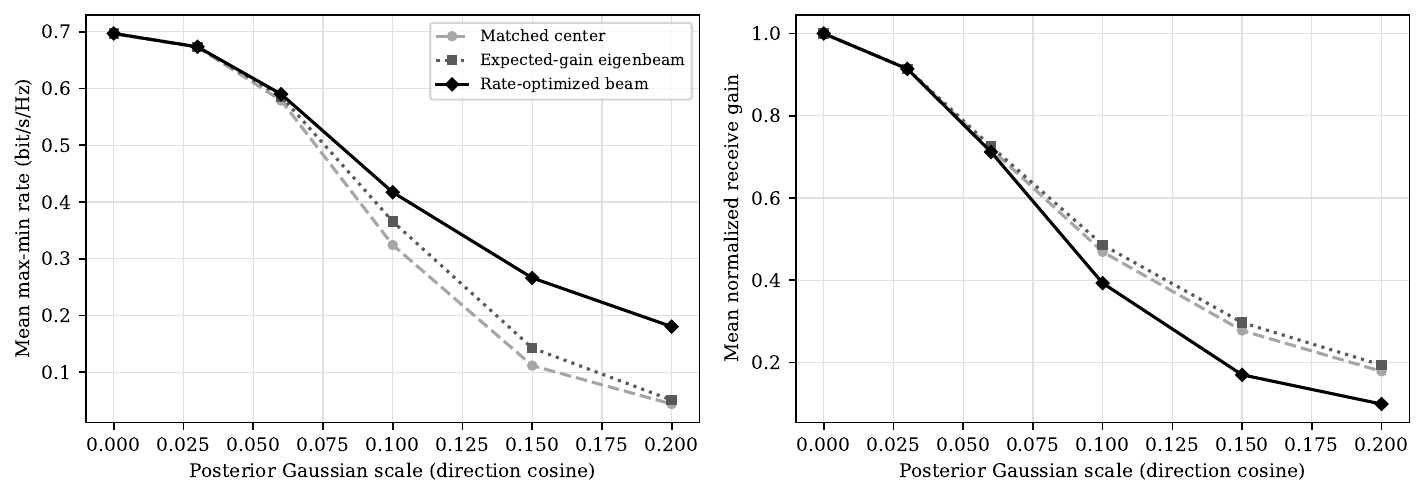}
\caption{Receive design under the conditional directional distribution. Direct optimization improves mean minimum rate even when the mean-gain eigenbeam has the higher average array gain. Each plotted rate uses an implementable rank-one beam and the separate evaluation scenarios.}
\label{fig:rate}
\end{figure}

\subsection{Distance from the finite-scenario global optimum}
Table~\ref{tab:bounds} compares the achieved design objective $L_J$ with the global upper bound $U$ from \eqref{eq:upper}. At $\sigma_\eta=0.10$, the feasible rank-one design achieves 0.41708656 bit/s/Hz, and every rank-one design on the same scenarios is bounded above by 0.41731418. Its distance from the global optimum is therefore at most 0.00022761 bit/s/Hz, or 0.05457\% of its achieved design value. At $\sigma_\eta=0.06$, the difference is only $1.44\times10^{-6}$ bit/s/Hz.

\begin{table}[tb]
\centering\small
\caption{Finite-scenario optimality bounds for the actual rank-one beams. All quantities are in bit/s/Hz and refer to the 512 design scenarios.}
\label{tab:bounds}
\begin{tabular}{rrrr}\toprule
$\sigma_\eta$ & Achieved $L_J$ & Global upper bound $U$ & $U-L_J$\\\midrule
0.03 & 0.673278016 & 0.673280737 & $2.721\times10^{-6}$\\
0.06 & 0.589733387 & 0.589734822 & $1.435\times10^{-6}$\\
0.10 & 0.417086564 & 0.417314177 & $2.276\times10^{-4}$\\
0.15 & 0.267036632 & 0.271962996 & $4.926\times10^{-3}$\\
0.20 & 0.182139883 & 0.191339746 & $9.200\times10^{-3}$\\\bottomrule
\end{tabular}
\end{table}

Figure~\ref{fig:bound} shows the upper-bound gap during the relaxed optimization. The final bound on the actual receive beam also includes rank-one extraction and refinement, as reported in Table~\ref{tab:bounds}. The larger differences at $\sigma_\eta\ge0.15$ leave more room for improvement despite the substantial rate gains over center matching. The calculations establish near-global optimality for the stated finite design problem at the smaller gaps; they do not establish an exact population optimum.

\begin{figure}[tb]
\centering\includegraphics[width=.78\linewidth]{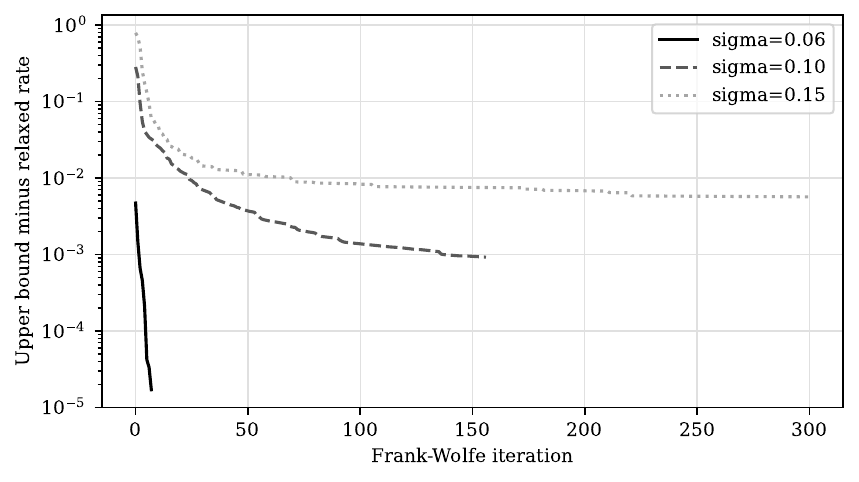}
\caption{Gap between the best global upper bound and the relaxed design objective during conditional-gradient optimization. The final rank-one gaps are given separately in Table~\ref{tab:bounds}.}
\label{fig:bound}
\end{figure}

The separate evaluation set estimates performance under the directional distribution, whereas $U$ bounds a different, fixed design sample average. Its numerical value therefore need not exceed an evaluation-set rate. Scrambled Sobol points within one set are dependent; independent-sample confidence intervals are not assigned to those individual points. The mathematical supporting bound is evaluated in floating-point arithmetic.

\subsection{Simulation outputs}
The accompanying source package includes both simulation programs, their common geometry and rate functions, the resulting trajectory statistics, the design and evaluation tables, the complex receive beams, and the convergence histories. The dynamic experiment uses seed 20260915. All tables and figures in this section are obtained from the included numerical outputs. The distinction between exact known-state or mean-gain optima, best linear prediction under a local moment model, and a bounded finite-scenario rate design is retained throughout the comparison.

\section{Conclusion}
Observation-axis optimization has distinct solutions for distinct physical objectives. With a known line-of-sight response, noise-weighted matching gives the global receive optimum for all RSMA decoding stages. With a directional posterior, a covariance eigenvector gives the exact mean-gain optimum, while equal-amplitude hardware is optimized by steering towards the Gaussian conditional mean. Direct optimization of mean fair rate generally produces a different beam. Its finite-scenario optimum can be bounded through the joint concavity of the RSMA water level.

The dynamic comparison shows that the principal improvement near ideal alignment comes from controlling prediction noise. A precision-weighted centered predictor removes more than half of the preceding direction error while retaining the same observations and transmit resources. Under broader directional uncertainty, direct rate optimization distributes receive gain more effectively across possible responses and materially improves the minimum user rate. PTFE preserves a common coherent reference throughout these weight changes. The resulting design separates physical prediction, observation transport and rate optimization, with an explicit objective and solution for each.

\appendix
\section{Supergradient Used in the Rate Optimization}
For scenario $j$, choose a common-rate bottleneck $b_j$ and an active smallest-private-rate group $A_j$ of size $m_j$ in \eqref{eq:subset}. A valid gain supergradient is
\begin{equation}
c_k^{(j)}=\frac1{m_j}\left[
\mathbf1_{k=b_j}f'_{c,k}(g_k^{(j)})+
\mathbf1_{k\in A_j}f'_{p,k}(g_k^{(j)})\right].
\end{equation}
Any active branch may be selected at a tie. The block matrix supergradient is
\begin{equation}
G_k=\frac1J\sum_j c_k^{(j)} A_k^{(j)}.
\end{equation}
It is Hermitian positive semidefinite because the rate derivatives are nonnegative. The conditional-gradient direction therefore requires only its leading eigenvector. The supporting inequality follows from the selected concave branch: it equals $F_J$ at the current point and dominates $F_J$ elsewhere.

\section{Relation to the Receiver-Side Angular Information Bound}
With successful common and intended-private decoding, the array mean contains known symbols and the undecoded streams have covariance $\beta a_ra_r^H+\sigma^2I$. For a centered $M_x\times M_y$ UPA and $L$ independent snapshots, the expected diagonal Fisher information of \cite{source} becomes
\begin{equation}
J_{pp}=\frac{L\pi^2N(M_p^2-1)}{6\sigma^2}
\left(\rho+\frac{N\beta^2}{\sigma^2+N\beta}\right),\qquad J_{uv}=0,
\end{equation}
where $\rho$ is the received common-plus-intended-private power and $\beta$ is the received power of the other private streams. The simulations use $J_{pp}^{-1}$ as the variance of an efficient-estimation approximation. This preserves the original signal-information model while making the sample count explicit. The optimal prediction and beam results require supplied error statistics; an actual estimator may have larger variance or non-Gaussian errors, in which case its measured conditional statistics should replace this approximation.

\end{document}